\documentclass[11pt]{amsart}

\usepackage[centertags]{amsmath}
\usepackage{amsthm,amsfonts}

\usepackage{amssymb}
\usepackage{epsfig}
\usepackage{amssymb,latexsym}
\usepackage{indentfirst}

\usepackage{mathrsfs}
\usepackage{amsthm}
\usepackage{amsmath}
\usepackage{amssymb}
\usepackage{amsfonts}
\usepackage{amsbsy}
\usepackage{appendix}
\usepackage{bbding}
\usepackage{tikz}
\usetikzlibrary{matrix,fit,shapes.misc,positioning}
\tikzset{%
  highlight/.style={rectangle,rounded corners,fill=red!15,draw,fill opacity=0.3,thick,inner sep=0pt}
}
\tikzset{%
  highlight1/.style={rectangle,rounded corners,fill=blue!15,draw,fill opacity=0.3,thick,inner sep=0pt}
}

\theoremstyle{plain}
\newtheorem{thm}{Theorem}[section]
\newtheorem{cor}[thm]{Corollary}
\newtheorem{lem}[thm]{Lemma}
\newtheorem{prop}[thm]{Proposition}

\theoremstyle{definition}
\newtheorem{defn}[thm]{Definition}
\newtheorem{rem}[thm]{Remark}
\newtheorem{ex}[thm]{Example}
\numberwithin{equation}{section}

\newcommand{\F}{{\mathbb F}}
\newcommand{\Tr}{{\rm Tr}}
\newcommand{\Span}{{\rm Span  }}
\newcommand{\lcm}{{\rm lcm }}

\begin{document}

\title[QCCD Codes]{Quasi-Cyclic Complementary Dual Codes}

\author{Cem G\"uner\.{i}}
\address{Sabanc{\i} University, FENS, 34956  \.{I}stanbul, Turkey}
\email{guneri@sabanciuniv.edu}

\author{Buket \"Ozkaya}

\address{ Dept. Comelec, T\'{e}l\'{e}com ParisTech, 46 rue Barrault, 75634 Paris Cedex 13, France}
\email{buket.ozkaya@telecom-paristech.fr}

\author{Patrick Sol\'e}
\address{ CNRS/LTCI, T\'{e}l\'{e}com ParisTech, 46 rue Barrault, 75634 Paris Cedex 13, France}
\email{sole@enst.fr}
\vspace{0.4cm}


\begin{abstract}
 LCD codes are linear codes that intersect with their dual trivially. Quasi-cyclic codes that are LCD are characterized and studied by using their concatenated structure.
 Some asymptotic results are derived. Hermitian LCD codes are introduced to that end and their cyclic subclass is characterized.
 Constructions of QCCD codes from codes over larger alphabets are given.
\end{abstract}
\maketitle

\section{Introduction}
Linear complementary codes (LCD) are linear codes that intersect with their dual trivially. This concept was introduced by Massey, following an Information Theoretic
motivation \cite{M}. It was rediscovered more recently in \cite{CG} from Boolean masking considerations, of interest in embarked cryptography.
The two main results so far in the theory of LCD codes is the characterization of the cyclic subclass \cite{M1} and the asymptotic goodness \cite{S2}. In the present
work we consider the more general subclass of quasi-cyclic complementary dual codes (QCCD). This was partially studied by \cite{EY} who put special attention to
one-generator family. We use the duality driven Chinese Remainder Theorem (CRT) decomposition championed in \cite{LS1,LS2} and more recently in \cite{GO2,GO1,GO}.
Since that decomposition was useful to study self-dual quasi-cyclic codes it is natural to consider it again for studying LCD codes.
While \cite{CG} only considers
binary codes, we have $q-$ary codes which is useful in several ways. In particular we generalize in the case of $q$ a square the cyclic subclass characterization
of \cite{M1}. We also use this extra flexibility for deriving new constructions of LCD codes by base field descent.

The material is organized as follows. Section 2 recalls the CRT set-up, on which Section 3 built to derive its asymptotic results.
Section 4 is dedicated to the Hermitian inner product. Section 5 considers special constructions, in particular from trace orthogonal bases.
\section{Background on Quasi-Cyclic Codes} \label{background}
In the whole paper $q$ denotes a prime power and $\F_q$ the finite field of that order.
A linear code over $\F_q$ is called a {\bf quasi-cyclic} (QC) code of index $\ell$ if it is closed under shifting codewords by $\ell$ units, and $\ell$ is the smallest positive integer with this property. So, cyclic codes amount to the special case $\ell=1$. It is well-known that the index of a QC code divides its length. So, we let $C$ be a QC code of length $m\ell$, index $\ell$ over $\F_q$. If we let $R:=\F_q[x]/\langle x^m-1 \rangle$, then the code $C$ can be viewed as an $R$-module in $R^\ell$ (\cite[Lemma 3.1]{LS1}).

As in \cite{LS1}, assume the following factorization into irreducible polynomials in $\F_q[x]$
\begin{equation}\label{factors}
x^m-1=g_1\cdots g_sh_1h_1^*\cdots h_th_t^*,
\end{equation}
where $g_i$'s are self-reciprocal and $h_j^*$ denotes the reciprocal of $h_j$. Let $\xi$ be a primitive $m^{th}$ root of unity over $\F_q$. Assume that $g_i(\xi^{u_i})=0$ and $h_j(\xi^{v_j})=0$ (for all $i,j$). Then we also have $h_j^*(\xi^{-v_j})=0$. By Chinese Remainder Theorem (CRT), $R$ decomposes as
{\small\begin{eqnarray} \label{CRT1}
\left(\bigoplus_{i=1}^s \F_q[x]/\langle g_i \rangle \right) \oplus \left(\bigoplus_{j=1}^t \Bigl( \F_q[x]/\langle h_j \rangle \oplus \F_q[x]/\langle h_j^* \rangle \Bigr) \right) & = & \left( \bigoplus_{i=1}^s \F_q(\xi^{u_i}) \right) \oplus \left( \bigoplus_{j=1}^t \Bigl( \F_q(\xi^{v_j}) \oplus \F_q(\xi^{-v_j}) \Bigr) \right). \nonumber
\end{eqnarray}}Again as in \cite{LS1}, we let $G_i=\F_q[x]/\langle g_i \rangle$, $H_j'=\F_q[x]/\langle h_j \rangle$ and $H_j''=\F_q[x]/\langle h_j^* \rangle$ for simplicity. The map that sends $a(x)\in R$ to the decomposition can be thought of as projections mod each irreducible factor or as follows:
$$a(x) \mapsto \left( \bigoplus_{i=1}^s a(\xi^{u_i}) \right) \oplus \left( \bigoplus_{j=1}^t \Bigl(a(\xi^{v_j}) \oplus a(\xi^{-v_j}) \Bigr) \right).$$
This decomposition naturally extends to $R^\ell$ and then $C\subset R^\ell$ decomposes as
\begin{equation}\label{CRT2}C=\left( \bigoplus_{i=1}^s C_i \right) \oplus \left( \bigoplus_{j=1}^t \Bigl( C_j' \oplus C_j'' \Bigr) \right),\end{equation}
where each component code is a length $\ell$ linear code over the base field ($G_i,H_j'$ or $H_j''$) it is defined (\cite[Section IV]{LS1}). Component
codes $C_i,C_j',C_j''$ are called the {\bf constituents} of $C$.

The constituents can be described in terms of the generators of $C$ (\cite[Lemma 2.1]{GO}). Namely, if $C$ is an $r$-generator QC code with generators
$$\{\bigl(a_{1,1}(x),\ldots ,a_{1,\ell}(x)\bigr),\ldots ,
\bigl(a_{r,1}(x),\ldots ,a_{r,\ell}(x)\bigr)\} \subset R^\ell ,$$
then
\begin{eqnarray} \label{consts}
C_i & = & \Span_{G_i}\bigl\{\bigl(a_{b,1}(\xi^{u_i}),\ldots
,a_{b,\ell}(\xi^{u_i})\bigr): 1\leq b \leq r \bigr\}, \ \mbox{for
$1\leq i \leq s$}, \nonumber \\
C_j' & = &  \Span_{H_j'}\bigl\{\bigl(a_{b,1}(\xi^{v_j}),\ldots
,a_{b,\ell}(\xi^{v_j})\bigr): 1\leq b \leq r \bigr\}, \ \mbox{for
$1\leq j \leq t$}, \\
C_j'' & = &  \Span_{H_j''}\bigl\{\bigl(a_{b,1}(\xi^{-v_j}),\ldots
,a_{b,\ell}(\xi^{-v_j})\bigr): 1\leq b \leq r \bigr\}, \ \mbox{for
$1\leq j \leq t$} .\nonumber
\end{eqnarray}

For $i\in \{1,\ldots ,s\}$, let $\theta_i$ be the generating primitive idempotent for the $q$-ary minimal cyclic code of length $m$, whose check polynomial is $g_i(x)$. This cyclic code $\langle \theta_i \rangle$ is isomorphic to the field $G_i$. Similarly, let $\theta_j'$ and $\theta_j''$ denote the primitive idempotent generators for the minimal cyclic codes which are isomorphic to the fields $H_j'$ and $H_j''$ (for $1\leq j \leq t$). By Jensen's work (\cite{J}), it was shown in \cite{GO} that the QC code $C$ above also has a concatenated decomposition
\begin{equation}\label{CRT4}C=\left( \bigoplus_{i=1}^s \langle \theta_i \rangle \Box
\mathfrak{C}_i \right) \oplus \left( \bigoplus_{j=1}^t \left( \langle \theta_j' \rangle \Box
\mathfrak{C}_j' \right) \oplus \left( \langle \theta_j'' \rangle \Box
\mathfrak{C}_j'' \right) \Bigr) \right),\end{equation}
where the outer codes $\mathfrak{C}_i,\mathfrak{C}_j',\mathfrak{C}_j''$ are length $\ell$ linear codes over $G_i,H_j',H_j''$, respectively,
and where $\Box$ denotes standard concatenation. More importantly, the outer codes and the constituents (in CRT decomposition) are the same (\cite[Theorem 4.1]{GO}): $\mathfrak{C}_i=C_i,\mathfrak{C}_j'=C_j',\mathfrak{C}_j''=C_j''$ for all $i,j$. The converse statement holds as well. Namely, if you start with arbitrary length $\ell$ outer codes (constituents) over the fields $G_i,H_j',H_j''$ and form the concatenation above, the resulting code is a length $m\ell$, index $\ell$ QC code over $\F_q$.

It was shown in \cite{LS2} that for a QC code $C$ with CRT decomposition as in (\ref{CRT2}), the (Euclidean) dual in $\F_q^{m\ell}$ is of the form
\begin{equation}\label{CRT3}
C^\bot=\left( \bigoplus_{i=1}^s C_i^{\bot_h} \right) \oplus \left( \bigoplus_{j=1}^t \Bigl( C_j''^{\bot_e} \oplus C_j'^{\bot_e} \Bigr) \right).
\end{equation}
Here, $\bot_h$ denotes the Hermitian dual on $G_i^\ell=\F_q(\xi^{u_i})^\ell$ (for all $1\leq i \leq s$). Definition of the Hermitian inner product of $\vec{c}=(c_1(\xi^{u_i}),\ldots ,c_\ell(\xi^{u_i})),\vec{d}=(d_1(\xi^{u_i}),\ldots ,d_\ell(\xi^{u_i}))\in \F_q(\xi^{u_i})^\ell$, where $c_b(x),d_b(x)\in R$ for all $1\leq b \leq \ell$, is as follows:
\begin{equation}\label{hermprod}
\langle \vec{c},\vec{d}\rangle := \sum_{b=1}^\ell c_b(\xi^{u_i})d_b(\xi^{-u_i}).
\end{equation}

A minimum distance bound on QC codes was given by Jensen (see \cite[Theorem 4 and page 792]{J}) based on the concatenated description of the code. This bound was later improved by a bound obtained in \cite{GO1} (see \cite{GO} for the proof of improvement), although just Jensen's bound, which is easier to state, will be enough for our purposes in Section \ref{results}.

\begin{thm}\label{jensen}
Let $C$ be a $q$-ary QC code of length $m\ell$, index $\ell$ with
the concatenated structure
$$C=\bigoplus\limits_{t=1}^{g} \langle \theta_{t} \rangle \Box \mathfrak{C}_{t}.$$
Assume that
$d(\mathfrak{C}_{1}) \leq d(\mathfrak{C}_{2}) \leq \cdots \leq
d(\mathfrak{C}_{g})$. Then we have
\begin{equation} \label{jensens concat bound}
d(C)\geq \min\limits_{1\leq e \leq g} \left\{d(\mathfrak{C}_{e})
d\Bigl(\langle \theta_{1} \rangle  \oplus \cdots \oplus \langle
\theta_{e} \rangle \Bigr) \right\}.
\end{equation}
\end{thm}

\section{QCCD Codes are Asymptotically Good}\label{results}
We start with a characterization of QCCD codes via their constituents.

\begin{thm} \label{CDcriteria}
Let $C$ be a $q$-ary QC code of length $m\ell$ and index $\ell$ whose CRT decomposition is as in (\ref{CRT2}) and (\ref{CRT4}). Then $C$ is QCCD code if and only if $C_i$ is Hermitian LCD for all $1\leq i \leq s$ and  $C_j' \cap C_j''^{\bot_e}=\{0\}$, $C_j'' \cap C_j'^{\bot_e}=\{0\}$, for all $1\leq j \leq t$.
\end{thm}

\begin{proof}
Immediate from the CRT decomposition of the dual code $C^{\bot}$ in (\ref{CRT3}).
\end{proof}

\begin{cor} \label{CDinstance}
Suppose that the CRT decomposition of $C$ is as in (\ref{CRT2}) with Euclidean LCD codes $C_j'=C_j''$ over $H_j'=H_j''$ (for all $1\leq j \leq t$) and Hermitian LCD codes $C_i$'s over $G_i$'s (for all $1\leq i \leq s$). Then $C$ is QCCD.
\end{cor}

\begin{proof}
Since each $C_j'$ is LCD, we have the necessary condition satisfied on the constituents (outer codes) defined over $H_j'$'s, for all $1\leq j\leq t$.
\end{proof}

Corollary \ref{CDinstance} suggests an easy construction of QCCD codes from LCD codes. In fact, we will use this idea to show the existence of good long QCCD codes in the next result.

\begin{thm}\label{asymptotics}
Let $q$ be a power of a prime and $m\geq 2$ be relatively prime to $q$ such that $q\equiv a$ mod $m$. Assume that no power of $a$ is congruent to $(m-1)$ mod $m$. Then there exists an asymptotically good sequence of $q$-ary QCCD codes where each QC code in the sequence has index length/$m$.
\end{thm}

\begin{proof}
Let $\xi$ be a primitive $m$th root of unity over $\F_q$. The condition on $m$ and $q$ guarantees that $\xi$ and $\xi^{-1}$ have distinct minimal polynomials $h'(x),h''(x)$ over $\F_q$. Let $H'=\F_q[x]/\langle h'(x) \rangle$ and $H''=\F_q[x]/\langle h''(x) \rangle$. Note that these fields are equal: $H'=\F_q(\xi)=\F_q(\xi^{-1})=H''$. Let us denote both by $H$. Let us denote the primitive idempotents corresponding to the $q$-ary length $m$ minimal cyclic codes with check polynomials $h'$ and $h''$ by $\theta'$ and $\theta''$ respectively. Let $(C_i)$ be an asymptotically good sequence of (Euclidean) LCD codes over $H$. Such a sequence exists by \cite{M} and \cite{S2}. Assume that each $C_i$ has parameters $[\ell_i,k_i,d_i]$. For each $i\geq 1$ define the $q$-ary QC code $D_i$ as
\begin{equation}\label{sequence}
D_i:=C_i\oplus C_i=\bigl( \langle \theta' \rangle \Box C_i \bigr)\oplus \bigl( \langle \theta'' \rangle \Box C_i \bigr) \subset H^{\ell_i}\oplus H^{\ell_i}.
\end{equation}
By Corollary \ref{CDinstance}, $D_i$ is a QCCD code. Index of $D_i$ is clearly $\ell_i$, which is the ratio of its length and $m$. If $e:=[H:\F_q]$, the length and the ($q$-ary) dimension of $D_i$ is $m\ell_i$ and $2ek_i$ respectively. By Jensen's bound (Theorem \ref{jensen}), the minimum distance of $D_i$ satisfies
$$d(D_i)\geq \min\{d(\langle \theta' \rangle)d_i,d(\langle \theta' \rangle \oplus \langle \theta'' \rangle)d_i \}\geq d(\langle \theta' \rangle \oplus \langle \theta'' \rangle)d_i.$$
For the sequence of QCCD codes $(D_i)$, the asymptotic rate is
$$R=\lim_{i\to \infty} \frac{2ek_i}{m\ell_i}=\frac{2e}{m}\lim_{i\to \infty} \frac{k_i}{\ell_i},$$
and this quantity is positive since $(C_i)$ is asymptotically good. For the asymptotic distance, we have
$$\delta=\lim_{i\to \infty} \frac{d(D_i)}{m\ell_i}\geq d(\langle \theta' \rangle \oplus \langle \theta'' \rangle) \lim_{i\to \infty} \frac{d_i}{\ell_i}.$$
Note again that $\delta$ is positive since $(C_i)$ is asymptotically good. Hence we have the desired result.
\end{proof}

Note that the proof of Theorem \ref{asymptotics} utilizes a pair of identical constituents which are defined over the ``associated" fields $H'$ and $H''$. Next, we will show that LCD codes with respect to Hermitian inner product are asymptotically good. This will allow us to obtain another asymptotic result on QCCD codes (see Thereom \ref{asymptotic-2-herm}). Let us note that our arguments will be very similar to those of Massey's in \cite{M}.

Let us first fix some notation. Let $f(x)\in \F_q[x]$ be a self-reciprocal irreducible factor of $x^m-1$ (i.e. $f$ is one of the $g_i$'s in (\ref{factors})). It is well-known that $\deg f=r$ is even and hence $G_f=\F_q[x]/\langle f(x) \rangle$ is an even degree extension of $\F_q$. Let us denote this field by $\F_Q$ (i.e. $Q=q^r$). In this case the Hermitian inner product on $\F_Q^\ell$ (cf. (\ref{hermprod})) is given by
\begin{equation}\label{herminnprod}
\langle (x_1,\ldots ,x_\ell),(y_1,\ldots, y_\ell ) \rangle =\sum_{i=1}^\ell x_iy_i^{\sqrt{Q}}.
\end{equation}
Let us denote, as it is usually done, raising to power $\sqrt{Q}$ (conjugation) operation by \ $\bar{}$ \ .

For an $\F_Q$-linear code $C\subset \F_Q^\ell$, it is easy to observe that $C^{\bot_h}=\left( C^{\sqrt{Q}}\right)^{\bot_e}=\left( \bar{C} \right)^{\bot_e}$, where $\bar{C}\subset \F_Q^\ell$ is another $\F_Q$-linear code whose codewords are obtained from those of $C$ by coordinate-wise conjugation. The following is clear.

\begin{lem}\label{hermit1}
If $C$ is an $[\ell,k]$ linear code over $\F_Q$ with a generating matrix $G$, then $\bar{C}$ is also an $[\ell,k]$ linear code over $\F_Q$ with a generating matrix $\bar{G}$. Moreover, $C$ is Hermitian LCD if and only if $C\cap (\bar{C})^{\bot_e} = \{0\}$. In other words, $\F_Q^\ell = C\oplus (\bar{C})^{\bot_e}$.
\end{lem}

The following gives a criteria for Hermitian LCD codes and is analogous to the result for Euclidean LCD codes in \cite[Proposition 1]{M}.

\begin{prop}\label{hermit2}
$C$ is a Hermitian LCD code if and only if the $k\times k$ matrix $G\bar{G}^T$ is nonsingular.
\end{prop}

\begin{proof}
Suppose $G\bar{G}^T$ is nonsingular. Consider the linear operator $L$ on $\F_Q^\ell$ defined by the $\ell \times \ell$ matrix $\bar{G}^T (G\bar{G}^T)^{-1}G$. Let $\vec{u}\in \F_Q^k$ and consider an arbitrary codeword $\vec{c}=\vec{u}G$ in $C$. Then,
$$L(\vec{c})=\vec{u}G\bar{G}^T (G\bar{G}^T)^{-1}G=\vec{u}G=\vec{c}.$$
Hence, the image of $L$ contains $C$ and is of dimension $\geq k$. For $\vec{v}\in (\bar{C})^{\bot_e} \subset \F_Q^\ell$, we have $\vec{v}\bar{G}^T=\vec{0}$. Hence,
$$L(\vec{v})=\vec{v}\bar{G}^T (G\bar{G}^T)^{-1}G=\vec{0}.$$
Therefore, the kernel of $L$ contains $(\bar{C})^{\bot_e}$ and it is of dimension $\geq (\ell - k)$. As a result, $\mbox{Im}(L)=C$, $\mbox{Ker}(L)=(\bar{C})^{\bot_e}$ and an element $\vec{c}\in C\cap \bar{C}^{\bot_e}$ satisfies $\vec{c}=L(\vec{c})=\vec{0}$.

For the converse, suppose $G\bar{G}^T$ is singular. Then there exists a nonzero $\vec{u}\in F_Q^k$ such that $\vec{u}G\bar{G}^T=\vec{0}$. Note that $\vec{c}=\vec{u}G$ is a nonzero vector in $C$ and it satisfies $\vec{c}\bar{G}^T=\vec{0}$. Therefore $\vec{c}$ also belongs to $(\bar{C})^{\bot_e}$, which is a contradiction.
\end{proof}

The following result proves that Hermitian LCD codes are asymptotically good. It is the analogue of \cite[Propositions 2 and 3]{M}.

\begin{thm}\label{hermit3}
Let $\tilde{C}$ be an $[\ell , k]$ linear code  over $\F_Q$ with a systematic generator matrix $\tilde{G}=[I_k:P]$. There exists a Hermitian LCD code $C$ over $\F_Q$ with parameters $[2\ell-k,k]$ and with $d(C)\geq d(\tilde{C})$. Hence, Hermitian LCD codes are asymptotically good.
\end{thm}

\begin{proof}
Suppose that the characteristic of $\F_Q$ is 2. Let $G=[I_k:P:P]$ and $C$ be the $Q$-ary linear code generated by $G$.  We have
$$G\bar{G}^T=[I_k:P:P][I_k:\bar{P}:\bar{P}]^T=I_k+ P\bar{P}^T + P\bar{P}^T=I_k,$$
since characteristic is 2. Hence by Proposition \ref{hermit2}, $C$ is Hermitian LCD.

If the characteristic is not 2, let $Q=s^2$ and note that $s$ has to be odd. Therefore $2(s+1)$ divides $Q-1=(s-1)(s+1)$ and there exists an element $a\in \F_Q$ such that $a^{s+1}=-1$. Now let $G=[I_k:P:aP]$ and $C$ be the code with this generating matrix. Again $G\bar{G}^T=I_k$ and hence $C$ is Hermitian LCD.

It is clear in both cases that the dimension of $C$ is $k$ and its minimum distance is at least as big as that of $\tilde{C}$. Hence Hermitian LCD codes are asymptotically good (in any characteristic) since linear codes are asymptotically good.
\end{proof}

\begin{thm}\label{asymptotic-2-herm}
Let $q$ be a power of a prime and $m\geq 2$ be relatively prime to $q$ such that there exists $1\leq i <m$ and $j\geq 0$ such that $iq^j \equiv (m-i)$ mod $m$. Then there exists an asymptotically good sequence of $q$-ary QCCD codes where each QC code in the sequence has index length/$m$.
\end{thm}

\begin{proof}
Let $\xi$ be a primitive $m^{th}$ root of unity over $\F_q$. The condition on $m$ and $q$ guarantees that for some $i$, $\xi^i$ and $\xi^{-i}$ are in the same $q$-cyclotomic coset. In other words, $\xi^i$ and $\xi^{-i}$ share the same minimal polynomial $g(x)$ over $\F_q$. Hence, such a polynomial $g(x)\in \F_q[x]$ is a self-reciprocal irreducible factor of $x^m-1$.

Let $G=\F_q[x]/\langle g(x) \rangle$. Let us denote the primitive idempotent corresponding to the $q$-ary length $m$ minimal cyclic code with check polynomial $g$ by $\theta$. Let $(C_i)$ be an asymptotically good sequence of Hermitian LCD codes over $G$. Such a sequence exists by Theorem \ref{hermit3}. Assume that each $C_i$ has parameters $[\ell_i,k_i,d_i]$. For each $i\geq 1$ define the $q$-ary QC code $D_i$ as the QC code with one outer code:
\begin{equation}\label{sequence1}
D_i:=\langle \theta \rangle \Box C_i.
\end{equation}
If $e:=[G:\F_q]$, the length and the ($q$-ary) dimension of $D_i$ is $m\ell_i$ and $ek_i$ respectively. It is well-known that for a concatenated code as above,
$$d(\langle \theta \rangle \Box C_i )\geq d(\langle \theta \rangle)d(C_i).$$
Since $(C_i)$ is asymptotically good, $(D_i)$ is also asymptotically good.
\end{proof}

\begin{cor}\label{asymptotic-3-herm}
For any pair $q$ and $m$, which are relatively prime, there exists an asymptotically good sequence of QCCD codes over $\F_q$ where each QC code in the sequence has index length/$m$.
\end{cor}

\begin{proof}
Note that for a given $q$ and $m$, one of the conditions in Theorems \ref{asymptotics} or \ref{asymptotic-2-herm} must be satisfied. In other words, $x^m-1$ has either a self-reciprocal irreducible factor other than $(x-1)$, hence Theorem \ref{asymptotic-2-herm} can be used, or all irreducible factors other than $(x-1)$ come in pairs so that Theorem \ref{asymptotics} can be applied.
\end{proof}

\section{Cyclic LCD Codes with respect to the Hermitian Inner Product} \label{HermitianProduct}

Next, we want to investigate when cyclic codes are Hermitian LCD. For the Euclidean case, this question is answered affirmatively by Yang and Massey in \cite{YM}. Our result will also yield a result on multidimensional versions of cyclic and QC codes (see Theorem \ref{2Dcyclic} and Remark \ref{nD}).

Let $\F_Q$ be an even degree extension of $\F_q$ as before and equip $\F_Q^\ell$ with the Hermitian inner product (\ref{herminnprod}). We will set $\ell=\tilde{\ell}p^e$, where $p$ is the characteristic of the finite field and $p\nmid \tilde{\ell}$. Recall that the conjugation operation is denoted by \ $\bar{}$ \ and it raises elements of $\F_Q$  to $\sqrt{Q}$ power. For a polynomial $f(x)\in \F_Q[x]$, we will denote the conjugate polynomial by $\bar{f}(x)$, whose coefficients are conjugates of the relevant coefficients of $f(x)$. It is clear that if $C$ is a length $\ell$ cyclic code over $\F_Q$ with the generating polynomial $g(x)\in \F_Q[x]$, then $\bar{C}$ is a $Q$-ary cyclic code of length $\ell$ with the generating polynomial $\bar{g}(x)$. In particular the dimension of $\bar{C}$ is the same as the dimension of $C$.

For a polynomial $f(x)\in \F_Q[x]$ with a nonzero constant coefficient $f_0$, let us denote the monic reciprocal polynomial by $\tilde{f}(x)$. Recall that
$$\tilde{f}(x)= f_0^{-1}x^{\deg f}f(x^{-1}).$$
For a cyclic code $C$ as above, it is well-known that the (Euclidean) dual cyclic code has the generating polynomial $\tilde{h}(x)$, where $h(x)\in \F_Q[x]$ is the polynomial that satisfies $x^{\ell}-1=g(x)h(x)$. Note that the conjugate and the monic reciprocal of $x^\ell -1$ are again $x^\ell -1$. Hence, we have
$$x^{\ell}-1=\bar{g}(x)\bar{h}(x)=\tilde{\bar{g}}(x)\tilde{\bar{h}}(x).$$
In particular, we have
$$\bar{C}^{\bot_e}=\langle \tilde{\bar{h}}(x) \rangle .$$

The following generalizes the results in \cite{YM} to the Hermitian setting.
\begin{thm}\label{reverse1}
With the notation so far, $C=\langle g(x) \rangle$ is an LCD cyclic code with respect to Hermitian inner product if and only if $\gcd \left(g(x),\tilde{\bar{h}}(x)\right)=1$. In particular, $C$ is Hermitian LCD if and only if $g(x)$ is conjugate-self-reciprocal (i.e. $g(x)=\tilde{\bar{g}}(x)$) and all the irreducible factors of $g(x)$ have the same multiplicity in $g(x)$ and in $x^\ell -1$.
\end{thm}

\begin{proof}
It is observed above that
$$x^{\ell}-1=g(x)h(x)=\tilde{\bar{g}}(x)\tilde{\bar{h}}(x).$$
Intersection of two cyclic codes is again a cyclic code and we have
$$C\cap C^{\bot_h}=C\cap \bar{C}^{\bot_e}=\langle \lcm \left( g(x),\tilde{\bar{h}}(x)\right) \rangle.$$
This intersection is trivial (i.e. $C$ is Hermitian LCD) if and only if $\lcm \left( g(x),\tilde{\bar{h}}(x)\right) =x^\ell -1$. If $k$ is the degree of $g(x)$, then $\ell -k$ is the degree of $\tilde{\bar{h}}(x)$. Hence, for the equality involving the least common multiple to hold, we must have $\gcd \left(g(x),\tilde{\bar{h}}(x)\right)=1$. The remaining assertions in the statement easily follow.
\end{proof}

A length $\ell$ linear code over $\F_Q$ is called reversible if $(c_0,\ldots ,c_{\ell -1})\in C$ implies that $(c_{\ell -1},\ldots ,c_0)\in C$. It was shown in \cite{M1} that a cyclic code is reversible if and only if its generating polynomial is self-reciprocal. When $\gcd (\ell , Q)=1$, reversibility of $C$ is equivalent to $C$ being Euclidean LCD (\cite{YM}). Next, we prove the analogous result for cyclic LCD codes with respect to Hermitian inner product.

\begin{defn}\label{conjreverse}
A length $\ell$ linear code $C$ over $\F_Q$ is called conjugate-reversible if $(c_0,\ldots ,c_{\ell -1})\in C$ implies that $(\bar{c}_{\ell -1},\ldots ,\bar{c}_0)\in C$.
\end{defn}

As before, let $C=\langle g(x) \rangle$ be a length $\ell$ cyclic code over $\F_Q$ and $(c_0,\ldots ,c_{\ell -1})$ be a codeword of $C$. Note that the polynomial representation of this codeword in $\F_Q[x]/\langle x^\ell -1 \rangle$ is
$$c(x)=c_0+c_1x+\cdots +c_{\ell -1}x^{\ell -1}.$$
Let us denote the polynomial representation of the conjugate-reverse $(\bar{c}_{\ell -1},\ldots ,\bar{c}_0)$ of this codeword by $\bar{c}_R(x)$ and observe that
$$\bar{c}_R(x)=x^{\ell -1} \bar{c}(x^{-1}).$$
Assume that $\deg g =k$ and let $g_0$ be the nonzero constant coefficient of $g(x)$. Suppose $c(x)=g(x)u(x)$ for some polynomial $u(x)\in \F_Q[x]$ of degree less than $\ell -k$. Then we have
\begin{eqnarray*}
\bar{c}_R(x) & = & x^{\ell -1}\bar{g}(x^{-1})\bar{u}(x^{-1})\\
& = & \left(\bar{g}_0^{-1}x^k \bar{g}(x^{-1}) \right) \left(\bar{g}_0 x^{\ell - k-1} \bar{u}(x^{-1}) \right) \\
& = & \tilde{\bar{g}}(x) \left(\bar{g}_0 x^{\ell - k-1} \bar{u}(x^{-1}) \right).
\end{eqnarray*}
Note that the map which sends $u(x)$ to $\left(\bar{g}_0 x^{\ell - k-1} \bar{u}(x^{-1}) \right)$ is a bijection on the set of polynomials in $\F_Q[x]$ of degree less than $\ell -1$. This implies that the code
$$\bar{C}_R:=\left\{\bar{c}_R(x): c(x)\in C \right\}$$
is a cyclic code of length $\ell$ with the generating polynomial $\tilde{\bar{g}}(x)$. Hence, $C=\bar{C}_R$ if and only if $g(x)=\tilde{\bar{g}}(x)$. This proves the following characterization of Hermitian LCD cyclic codes.

\begin{thm}\label{hermit4}
A cyclic code over $\F_Q$ of length $\ell$, where $\gcd(\ell , Q)=1$, is Hermitian LCD if and only if $C$ is conjugate-reversible.
\end{thm}

Characterizations of Euclidean LCD cyclic codes (\cite{YM}) and Hermitian LCD cyclic codes (Theorem \ref{hermit4}) yield a characterization of LCD 2D cyclic codes. Let us recall that a cyclic $R$-linear code of length $\ell$ (i.e. an $R$-submodule of $R^\ell$ which is closed under cyclic shift) is called a 2D cyclic code. So, when viewed as $\F_q$-linear codes, 2D cyclic codes are length $m\ell$, index $\ell$ QC codes with extra structure (see \cite{GO2} for further information). Therefore one can also decompose a 2D cyclic code into constituents (outer codes). It has been shown that the constituents of a 2D cyclic code ($C_i,C_j',C_j''$ in (\ref{CRT2})) are cyclic codes of length $\ell$ over their fields of definition (\cite[Theorem 3.5]{GO2}). Hence, Theorem \ref{CDcriteria} and Corollary \ref{CDinstance} apply to 2D cyclic codes. Combining these with the results on Euclidean and Hermitian LCD cyclic codes, we obtain the following.

\begin{thm}\label{2Dcyclic}
Let $C$ be a 2D cyclic code with a decomposition as in (\ref{CRT2}) and assume that both $m$ and $\ell$ are relatively prime to $q$. Assume that $C_i$ is a conjugate-reversible code over $G_i$ for all $1\leq i \leq s$ and $C_j'=C_j''$ are reversible codes over $H_j'=H_j''$ for all $1\leq j \leq t$. Then $C$ is an LCD 2D cyclic code.
\end{thm}

\begin{rem}\label{nD}
Multidimensional ($n$D) cyclic codes have been studied in the literature also for $n>2$ (see \cite{GO08,SH}). Multidimensional versions of QC codes were recently introduced (\cite{O}). Although their original definitions would require lengthy introduction and notation, both classes of codes can be viewed as QC codes and characterized by their constituents. Recall from above that a 2D cyclic code is a QC code whose constituents are not just linear but cyclic. Recursively, we can define an $n$D cyclic code as a QC code whose constituents are $(n-1)$D cyclic codes, for all $n\geq 2$ (cf. \cite[Theorem 4.3]{GO}). Analogously, a quasi-2D-cyclic code (Q2DC) is a QC code whose constituents are also QC and recursively, one can define a Q$n$DC code as a code with Q$(n-1)$DC constituents (\cite[Theorem 2.3.1]{O}). In short, one can say that an $n$D cyclic codes are generated from cyclic codes and Q$n$DC codes are generated from linear codes, through suitable concatenations (as described in Section \ref{background}). Hence, a natural problem to attempt is to formulate conditions on LCD codes in both families, such as the one for 2D cyclic codes in Theorem \ref{2Dcyclic}.
\end{rem}

\section{Constructions and Examples} \label{examples}

Direct sum of two LCD codes is again an LCD code. According to \cite{CG}, this is the only known construction of an LCD code out of other LCD codes. We will provide constructions and examples of LCD codes in this section. We use the computational algebra system Magma \cite{BCP} in the examples.

\subsection{Construction from the Constituents}
Theorem \ref{CDcriteria} and Corolary \ref{CDinstance} describe the characterization of QCCD codes based on Euclidean and Hermitian LCD codes.
These characterizations have already been used for showing the existence of asymptotically good QCCD codes (Section \ref{results}).
Another application is seen in Theorem \ref{2Dcyclic} where LCD 2D cyclic codes are described.
Here, we will provide examples of LCD codes based on the construction from constituents.
Random search in Magma is carried out for linear codes over various
extensions of $\F_2$, which satisfy the condition in Theorem 3.1.
The numbers in parentheses are the optimal distances. All the codes
are
of dimension $m$.\\

$$\begin{array}{|c|c|c|c|c|c|c|c|c|}
  \hline
  \mathbf{m/\ell} & \mathbf{3} & \mathbf{5} & \mathbf{7} & \mathbf{9} & \mathbf{11} & \mathbf{13} & \mathbf{15} & \mathbf{17} \\\hline
  \mathbf{2} & 2(3) & 3(4) & 4(4) & 4(6) & 5(7) & 6(7) & 6(8) & 6(8) \\\hline
  \mathbf{3} & 3(4) & 5(7) & 7(8) & 8(10) & 9(12) & 10(12) & 11(14) &  \\\hline
  \mathbf{5} & 7(8) & 10(12) & 13(16) & 15(18) &  &  &  &  \\\hline
  \mathbf{7} & 10(12) &  &  &  &  &  &  &  \\
  \hline
\end{array}$$\\

\subsection{Double Circulant Codes}

Let $R=\F_q[x]/\langle x^m-1\rangle$ as before and let $C=\left\langle \bigl(1,a(x)\bigr) \right\rangle \subset R^2$ be a systematic double circulant code in this section.

\begin{thm}\label{DC}
$C=\left\langle \bigl(1,a(x)\bigr) \right\rangle$ is LCD if and only if $\gcd \left(a(x)a(x^{m-1})+1,x^m-1\right)=1$.
\end{thm}

\begin{proof}
Let $\xi$ denote a primitive $m^{th}$ root of unity and assume that $x^m-1$ factors as in (\ref{factors}). Inherit all the notation from Section \ref{background}. In particular, by (\ref{consts}), the constituents of $C$ are
\begin{eqnarray} \label{consts2}
C_i & = & \Span_{G_i}\bigl\{\bigl(1,a(\xi^{u_i})\bigr): 1\leq b \leq r \bigr\}, \ \mbox{for
$1\leq i \leq s$}, \nonumber \\
C_j' & = &  \Span_{H_j'}\bigl\{\bigl(1,a(\xi^{v_j})\bigr): 1\leq b \leq r \bigr\}, \ \mbox{for
$1\leq j \leq t$}, \\
C_j'' & = &  \Span_{H_j''}\bigl\{\bigl(1,a(\xi^{-v_j})\bigr): 1\leq b \leq r \bigr\}, \ \mbox{for
$1\leq j \leq t$} .\nonumber
\end{eqnarray}
Note that each constituent is a 1-dimensional space over the two dimensional ambient space it lives in. Hence, duals of constituents are all 1-dimensional too. Therefore intersections we want to check (cf. Theorem \ref{CDcriteria}) are either trivial or 1-dimensional.

$C_i\cap C_i^{\perp_h}\not= \{0\}$ if and only if $C_i=C_i^{\perp_h}$, which is equivalent to
$$1+a(\xi^{u_i})a(\xi^{-u_i})=0 \ \ \ \mbox{(cf. (\ref{hermprod}))}.$$
On the other hand, $C_j' \cap C_j''^{\perp_e}\not=0$ if and only if $C_j'=C_j''^{\perp_e}$, which is equivalent to
$$1+a(\xi^{v_j})a(\xi^{-v_j})=0.$$
The last intersection $C_j'' \cap C_j'^{\perp_e}$ does not bring a new condition. Hence, being LCD for $C$ is equivalent to the polynomial $a(x)a(x^{m-1})+1$ not vanishing at any $m^{th}$ root of unity.
\end{proof}

The following table presents the best possible distances for double
circulant binary LCD codes $C=\left\langle \bigl(1,a(x)\bigr)
\right\rangle \subset R^2$, where $C$ is QC of length $2m$,
dimension $m$ and index 2. The search is done in Magma for random
$a(x) \in R$ satisfying the condition in Theorem 5.5 and the ones
marked with ``*"
are optimal or with best-known parameters.\\

$$\begin{array}{|c|c|c|c|c|c|c|c|c|}
  \hline
 \mathbf{ m} & \mathbf{3} & \mathbf{5} & \mathbf{7} & \mathbf{9} & \mathbf{11} & \mathbf{13} & \mathbf{15} & \mathbf{17} \\\hline
  \mathbf{d} & 1 & 3 & 4* & 3 & 6 & 7* & 5 & 8* \\\hline
  \mathbf{d*} & 3 & 4 & 4 & 6 & 7 & 7 & 8 & 8 \\
  \hline
\end{array}$$\\

For instance, using the code denoted in Magma notation by
{\tt QuasiCyclicCode(10,[1,1+x+$x^3$]) }
we obtain a $[10,5,3]$ binary code. Note that the best linear $[10,5]$ codes have parameters $[10,5,4].$ But the one obtained from the Magma BKLC command is not LCD.


\subsection{Subfield Construction}

Let $\mathcal{B}=\{\beta_1,\ldots ,\beta_\ell\}$ be a self-dual basis of $\F_{q^\ell}$ over $\F_q$. If $\Tr$ denotes the trace map from $\F_{q^\ell}$ to $\F_q$, this means that
$$\Tr(\beta_i \beta_j )= \left\{ \begin{array}{ll} 0 & \mbox{if $i\not= j$} \\ 1 & \mbox{if $i=j$} \end{array}\right. .$$
Note that a self-dual basis exists if and only if $q$ is even or both $q$ and $\ell$ are odd (\cite{SL}).

For an element $x\in \F_{q^\ell}$, let us denote the coordinates relative to $\mathcal{B}$ as $\vec{x}_{\mathcal{B}}=(x_1,\ldots ,x_\ell)\in \F_q^\ell$ (i.e. $x=\sum_i x_i\beta_i$). For $x,y\in \F_{q^\ell}$, we have
\begin{equation}\label{id-1}
\Tr(xy)=\vec{x}_{\mathcal{B}}\cdot \vec{y}_{\mathcal{B}},
\end{equation}
where the operation on the right hand side is the Euclidean inner product on $\F_q^\ell$.

Consider the $\F_q$-linear isomorphism
\begin{eqnarray*}
\varphi : \F_{q^\ell}^m & \longrightarrow \F_{q}^{m\ell} \\
(c_1,\ldots ,c_m) & \mapsto & \left(\begin{array}{ccc} c_{11} & \cdots & c_{1\ell} \\ \vdots & & \vdots \\  c_{m1} & \cdots & c_{m\ell}  \end{array} \right),
\end{eqnarray*}
where the $i^{th}$ row in the image consists of the $\mathcal{B}$-coordinates of $c_i \in \F_{q^\ell}$ (for all $1\leq i \leq m$). For $\vec{a}=(a_1,\ldots ,a_m), \vec{b}=(b_1,\ldots ,b_m) \in \F_{q^\ell}^m$, we observe (using (\ref{id-1})) that
\begin{equation}\label{id-2}
\Tr(\vec{a}\cdot \vec{b})=\varphi(\vec{a}) \cdot \varphi (\vec{b}) .
\end{equation}
Here, the operations on the left and right sides of the equation are Euclidean inner products in $\F_{q^\ell}^m$ and $\F_q^{m\ell}$, respectively.

\begin{thm}\label{subfield}
A linear code $C\subset \F_{q^\ell}^m$ is LCD if and only if the linear code $\varphi(C) \subset \F_{q}^{m\ell}$ is LCD.
\end{thm}

\begin{proof}
We have $\varphi (C\cap C^{\bot_e})=\varphi (C) \cap \varphi (C^{\bot_e})$, since $\varphi$ is an isomorphism. We will show that $\varphi (C^{\bot_e})=\varphi (C)^{\bot_e}$, which will suffice for the proof.

If $\vec{x}\in C^{\bot_e}$, then for any $\vec{c}\in C$ we have
$$0=\Tr (\vec{x}\cdot \vec{c})=\varphi (\vec{x})\cdot \varphi (\vec{c}) \hspace{0.6cm}  \mbox{(cf. (\ref{id-2}))} .$$
Hence $\varphi (C^{\bot_e}) \subseteq \varphi(C)^{\bot_e}$. On the other hand, if $\dim_{\F_{q^\ell}}(C)=k$, then
$$\dim_{\F_q}(\varphi(C^{\bot_e}))=\dim_{\F_q}(C^{\bot_e})=\ell (m-k),$$
and
$$\dim_{\F_q}(\varphi(C)^{\bot_e})=m\ell -\dim_{\F_q}(\varphi(C))=m\ell -\dim_{\F_q}(C)=m\ell -mk.$$
Hence, $\varphi (C^{\bot_e})= \varphi(C)^{\bot_e}$.
\end{proof}

\begin{rem}
Note that if $C$ is a cyclic code of length $m$ over $\F_{q^\ell}$, then $\varphi(C)$ is a linear code over $\F_q$ of length $m\ell$, which is closed under shift of codewords by $\ell$ units.
\end{rem}

\begin{ex}
Let $C$ be the cyclic quaternary LCD code of parameters $[15,11,3],$ with generator polynomial $g=f \tilde{f},$ where
$f(x)=x^2+x+w,$ $\tilde{f}(x)=x^2+w^2x+w^2,$ the monic reciprocal of $f$ and $\F_4=\F_2(w).$ With these notations we have that $\varphi(C)$ is a binary  $[30,22,3]$ code that is LCD. There exists a $[30,22,4]$ code but we cannot find an LCD one.
\end{ex}



\end{document}